\documentclass[conference]{IEEEtran}
\usepackage{cite}
\usepackage[pdftex]{graphicx}
\usepackage{amsmath, amsfonts, amssymb, amsthm}
\usepackage{array}
\usepackage[switch]{lineno}
\usepackage{bm}
\usepackage{color}
\usepackage{epsfig}
\usepackage{mathrsfs}
\usepackage{epstopdf}
\usepackage{stmaryrd}
\usepackage{color}
\usepackage{subcaption}
\usepackage{verbatim}

\usepackage[ruled,linesnumbered]{algorithm2e}

\newcommand{\R}{\mathbb{R}}  
\newcommand{\C}{\mathbb{C}} 
\newcommand{\E}{\mathbb{E}}

\newtheorem{theorem}{Theorem}

\newtheorem{definition}{Definition}
\newtheorem{assumption}{Assumption}

\begin{document}
\title{Over-the-Air Decentralized Federated Learning}

\author{
	\IEEEauthorblockN{Yandong Shi\IEEEauthorrefmark{1}\IEEEauthorrefmark{2}\IEEEauthorrefmark{3}, Yong Zhou\IEEEauthorrefmark{1} , and Yuanming Shi\IEEEauthorrefmark{1}}
	\IEEEauthorblockA{\IEEEauthorrefmark{1}School of Information Science and Technology, ShanghaiTech University, Shanghai 201210, China}
	\IEEEauthorblockA{\IEEEauthorrefmark{2}Shanghai Institute of Microsystem and Information Technology, Chinese Academy of Sciences, China }
	\IEEEauthorblockA{\IEEEauthorrefmark{3}University of Chinese Academy of Sciences, Beijing 100049, China \\
		Email: \{shiyd, zhouyong, shiym\}@shanghaitech.edu.cn}
	
}
\maketitle	
\let\thefootnote\relax\footnotetext{This work was supported in part by the National Natural Science Foundation of China (NSFC) under grant 62001294.}

\begin{abstract}
In this paper, we consider decentralized federated learning (FL) over wireless networks, where over-the-air computation (AirComp) is adopted to facilitate the local model consensus in a device-to-device (D2D) communication manner.
However, the AirComp-based consensus phase brings the additive noise in each algorithm iterate and the consensus needs to be robust to wireless network topology changes, which introduce a coupled and novel challenge of establishing the convergence for wireless decentralized FL algorithm.
To facilitate consensus phase, we propose an AirComp-based DSGD with gradient tracking and variance reduction (DSGT-VR) algorithm, where both precoding and decoding strategies are developed for D2D communication.
Furthermore, we prove that the proposed algorithm converges linearly and establish the optimality gap for strongly convex and smooth loss functions, taking into account the channel fading and noise.
The theoretical result shows that the additional error bound in the optimality gap depends on the number of devices.
Extensive simulations verify the theoretical results and show that the proposed algorithm outperforms other benchmark decentralized FL algorithms over wireless networks.
\end{abstract}

\section{Introduction}\label{introduction}
Federated learning (FL), as a novel type of distributed machine learning \cite{mcmahan2017communication}, has received a growing interest recently, both from academia and industry.
With FL, many mobile devices collaboratively train a global model under the orchestration of a parameter server (PS), while keeping the training data unmoved.
The PS only receives the local model rather than raw data from mobile devices, thereby significantly reducing the overhead of data collection and preserving additional privacy \cite{rah2020survey, yuan2020edgeai}.
FL has enabled a wide range of applications in industrial Internet of Things (IIoT) \cite{lu2020iiot, pham2021fusion, savazzi2021opportunities}, e.g., autonomous driving and collaborative robotics, in which high levels of security, privacy and reliability are demanded.
In future wireless networks, the driving applications of FL including resource management, channel estimation and signal detection, edge caching, and computation offloading \cite{yang2021federated ,lim2020fl}, which also poses unique challenges including statistical heterogeneity, communication cost and resource allocation \cite{nik2020fl}.
It thus becomes critical to implement FL for the future 6G wireless networks \cite{kb20196g, yang2020flris}.

Due to limited radio resources, communication is a fundamental bottleneck in FL over wireless networks \cite{li2021delay, wang2020federated}. 
The digital \cite{chen2020conver, chen2020joint, wen2020joint} and analog \cite{yang2020over, sery2020flmac, mo2020fl} transmission schemes have been proposed to support FL over wireless networks.
In particular, the authors in \cite{wen2020joint} partitioned a large-scale learning task over multiple resource-constrained edge devices, where joint parameter-and-bandwidth allocation was proposed to reduce the total learning-and-communication latency.
In \cite{chen2020conver, chen2020joint}, a joint learning, wireless resource allocation, and user selection problem was formulated to minimize FL training loss and FL convergence time.
However, the orthogonal channels are required in digital FL systems which suffer from hugely demanding bandwidth especially when number of edge devices is large \cite{sery2020flmac}.
As for analog transmission scheme, over-the-air computation (AirComp) was proposed to support fast model aggregation for FL by exploiting the superposition property of multi-access channel (MAC), where multiple edge devices share the same frequency channel \cite{yang2020over, sery2020flmac, mo2020fl}.
Specifically, a joint device selection and beamforming design was proposed to improve the statistical learning performance for FL \cite{yang2020over}.
The authors in \cite{sery2020flmac, mo2020fl} focused on designing a stochastic gradient descent (SGD) based algorithm over MAC and investigate the effect of channel fading and noise on the convergence analysis of FL.

The aforementioned studies require a central PS to orchestrate the training process. 
However, in some application scenarios, e.g., cooperative driving and robotics \cite{savazzi2021opportunities}, a central PS may not always be available and reliable when the number of edge devices is large \cite{lian2017can}.
The centralized FL also faces straggler’s dilemma \cite{cai2020leter} due to the heterogeneity of edge devices, i.e., the FL training speed is limited by the devices with slowest computation and worst channel conditions.
To overcome these challenges, device-to-device (D2D) communications based decentralized FL \cite{ozfatura2020decentralized, xing2020decentralized, sa2020de} was proposed, where each device only communicates with its neighbors.
In particular, \cite{sa2020de} considered digital transmission schemes in a joint learning and network simulation framework, to quantify the effects of model pruning, quantization and physical layer constraints for decentralized FL.
Due to limited wireless bandwidth resources, the authors in \cite{ozfatura2020decentralized, xing2020decentralized} proposed a decentralized stochastic gradient descent (DSGD) algorithm to improve the convergence performance in decentralized FL, where AirComp based D2D communication was developed to facilitate the consensus phase.
However, the AirComp-based consensus phase brings the channel fading and additive noise in each algorithm iterate.
On the other hand, the neighborhood weighted average involving the information of network topology performs the consensus phase. 
The decentralized FL algorithm thus needs to be robust against changes in network topology, otherwise it may not converge \cite{lian2017can}.
These introduce a coupled and unique challenge of establishing the convergence of decentralized FL algorithm over wireless networks. 

In this paper, we shall consider a decentralized FL model over wireless networks, where no central PS exists to orchestrate the training process.
To investigate and facilitate consensus phase over wireless network, we propose an AirComp-based DSGT-VR algorithm in decentralized FL, where both precoding and decoding strategies at devices are developed to guarantee algorithm convergence.
In addition, the gradient tracking and variance reduction techniques are involved in the proposed algorithm, which can further improve algorithm convergence performance.
We analyze the performance of AirComp-based DSGT-VR algorithm theoretically by introducing the auxiliary variable in consensus phase, i.e., the mean of all local parameters.
For strongly convex and smooth local loss functions, we prove that the proposed algorithm converges linearly and also establish the optimality gap, taking into account the channel fading and noise.
In addition, the convergence result shows that the additive error bound in the optimality gap depends on the number of devices.
Numerical experiments are conducted to validate the theoretical analysis and demonstrate the superior performance of the proposed algorithm over wireless networks.
 
\subsubsection*{Notations}
Throughout this paper, we denote the cardinality of set $A$ by $|A|$, $\ell_2$-norm of vector $\bm{x}$ by $\Vert \bm{x} \Vert_2$ and the second largest
singular value of matrix $\bm{W}$ by $\interleave \bm{W} \interleave$.

\section{System Model and Problem Formulation}
\subsection{System Model}\label{model}
As shown in Fig. \ref{fig: system model}, we consider a federated learning system supported by a  decentralized wireless communication network in decentralized setting where no central PS exists to coordinate training process of all edge devices \cite{kairouz2019advances}.
We denote $\mathcal{N} = \{1,...,N\}$ as the set of devices and $\mathcal{G} = (\mathcal{N}, \mathcal{E})$ as an undirected graph with vertex set $\mathcal{N}$ and edge set $\mathcal{E}$ that represents the set of communication links.
The set of connected neighbors of device $i$ is denoted as $\mathcal{N}_i= \{ j|(i,j) \in \mathcal{E} \} $.
Each device $n \in \mathcal{N}$ has its own data set $\mathcal{D}_n$ and all devices aim to collaboratively learn a common machine learning model by communicating with each other through device-to-device communication link over graph $\mathcal{G}$.

The goal of this paper is to learn a global model by tackling the distributed stochastic optimization problem
\begin{align}\label{eq: problem}
\underset{\bm{\theta} \in \R^d }{\operatorname { minimize }}~~ F(\bm{\theta})\triangleq \frac{1}{N} \sum_{i=1}^{N} f_{i}\left(\bm{\theta}\right) ,
\vspace{-1mm}
\end{align}
where $F(\bm{\theta})$ is the global loss function and $f_i(\bm{\theta}) = \frac{1}{|\mathcal{D}_i|} \sum_{\xi \in \mathcal{D}_i} f_{i, \xi}(\bm{\theta})$ is the local loss function at device $i$.
In particular, $f_{i, \xi}(\bm{\theta})$ is the loss for the parameter $\bm{\theta} \in \R^d$ on a data sample $\xi$ at device $i$.
In the distributed learning process, each device $i \in \mathcal{N}$ has a local parameter vector $\bm{\theta}_i^{t}$ that approximates the
solution of problem (\ref{eq: problem}) at the $t$-th iteration.

\subsection{DSGD with Gradient Tracking and Variance Reduction}
When the devices communicate over error-free orthogonal links, the commonly adopted method for solving problem (\ref{eq: problem}) is DSGD \cite{koloskova2020unified}.
We denote $\bm{W}$ as the mixing matrix encoding the network structure at iteration $t$.
\begin{definition}[Mixing matrix]
Mixing matrix $\bm{W} $ is a symmetric doubly stochastic matrix whose $ij$-th entry $w_{ij} > 0$ indicates that device $i$ and $j$ are connected, 
\begin{align}
\bm{W} \in \mathcal{W} ,\bm{W} = \bm{W}^T,\bm{W}\bm{1} = \bm{1}, \bm{1}^T\bm{W} = \bm{1}^T,
\end{align}
where $\mathcal{W} = \{\bm{W} \in [0,1]^{N \times N} | w_{ij} = 0 ~\text{if} (i,j) \notin \mathcal{E} ~\text{and}~ i \neq j \}$.
\end{definition}

With the DSGD algorithm, the information exchange can only occur between connected devices. 
The algorithm for each device $i \in \mathcal{N}$ at iteration $t$ consists of two phases:
\begin{enumerate}
\item Computation phase: each device computes a stochastic gradient based on a random data sample $\xi_i^{t}$ and then performs stochastic gradient updates,
\begin{align}
\bm{\theta}_i^{t+ \frac{1}{2}} = \bm{\theta}_i^{t} - \alpha_t \nabla f_{i, \xi_i^{t}}(\bm{\theta}_i^t).
\end{align}
\item Consensus phase: devices average their local model based on mixing matrix,
\begin{align}\label{conse}
\bm{\theta}_i^{t+1} = \sum_{j \in \mathcal{N}_i} w_{ij} \bm{\theta}_j^{t+ \frac{1}{2}}.
\end{align}
\end{enumerate}
\begin{figure}[tbp]
        \centering
        \includegraphics[width=0.7\linewidth]{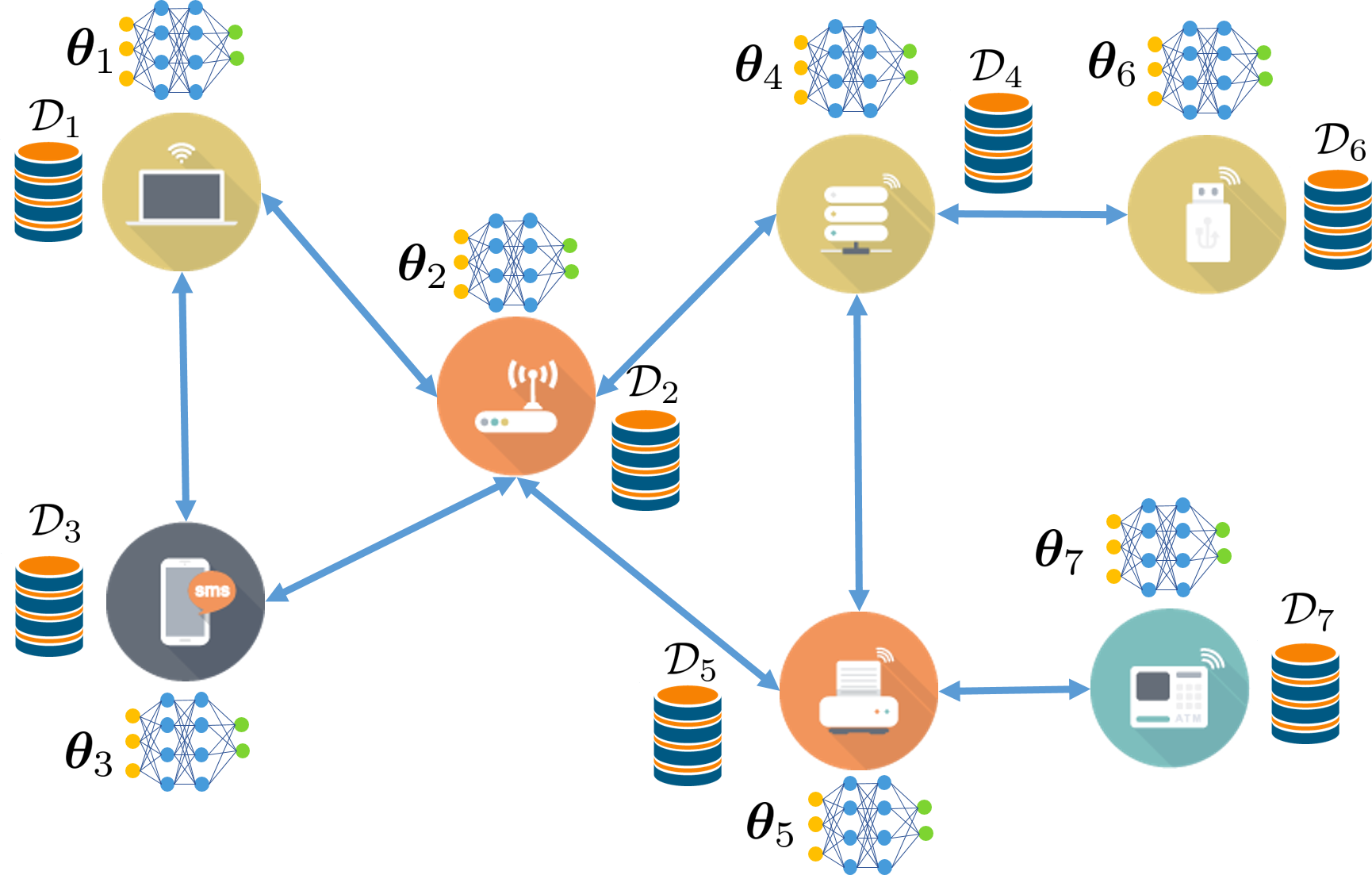}
        \caption{Illustration of the decentralized FL model over wireless networks consisting of $7$ devices.}
        \label{fig: system model}
        \vspace{-0.7cm}
\end{figure}

However, DSGD decays linearly with a fixed step size $\alpha_t \equiv \alpha$, but performs inexact convergence (the steady error of DSGD has an additional bias).
Although properly reducing $\alpha_t$ enables exact convergence (each $\bm{\theta}_n^t$ converges to the same exact solution), DSGD convergences slowly both in practice and theory \cite{xin2020magdso}.
With the help of dynamic average-consensus, the DSGD with gradient tracking (DSGT) removes the bias which characterizes the difference between local loss function $f_i$ in DSGD, thereby achieving a much lower steady-state error when the data sample across devices are largely heterogeneous \cite{pu2020distributed}.
And by adding variance reduction to DSGT, \cite{xin2020magdso} shows that the DSGT with variance reduction (DSGT-VR) leads to an exact linear convergence with a constant step-size.

In DSGT-VR, each device $i$ maintains a gradient table $\{ \nabla f_{i,j}(\hat{\bm{\theta}}_{i,j}) \}_{j=1}^{|\mathcal{D}_i|}$, where $\hat{\bm{\theta}}_{i,j}$ denotes the most recent parameter when $f_{i,j}$ was computed, to store all local gradients.
Then at iteration $t$, each device $i$ randomly chooses a data sample $\xi_i^t$ and computes the unbiased gradient $\bm{g}_i^t$ as follows
\begin{align}\label{unbag}
\bm{g}_i^t = \nabla f_{i,\xi_i^t}(\bm{\theta}_i^t) - \nabla f_{i,\xi_i^t}(\hat{\bm{\theta}}_{i, \xi_i^t}) + \frac{1}{|\mathcal{D}_i|}\sum_{j=1}^{|\mathcal{D}_i|} \nabla f_{i,j}(\hat{\bm{\theta}}_{i,j}) .
\end{align}
Next, we replace the recent gradient $\nabla f_{i,\xi_i^t}(\bm{\theta}_i^t)$ with $\nabla f_{i,\xi_i^t}(\hat{\bm{\theta}}_{i, \xi_i^t})$ in gradient table.
Last, the gradient estimator $ \bm{d}_i^{t+ 1}$ can be computed by using gradient tracking technology,
\begin{align}\label{gradient_est}
\bm{d}_i^{t+ 1} = \sum_{j \in \mathcal{N}_i} w_{ij}\bm{d}_i^{t} + \bm{g}_{i}^{t+1} - \bm{g}_{i}^{t}.
\end{align}

\subsection{Communication Model}
In this section, we focus on the design of communication model for AirComp-based DSGT-VR algorithm in decentralized setting.
The designed communication model consists of the following two parts:
\begin{itemize}
\item \textbf{Scheduling}: To cope with wireless interference, different devices should be scheduled to communicate in different transmission blocks.
\item \textbf{Transmission}: The transmission strategies are designed to support consensus phase over wireless networks.
\end{itemize}
\subsubsection{Scheduling}
To improve the communication performance, we consider the D2D communication \cite{te2014d2d} in this paper.
Since the precoding strategy (\ref{precoding}) is designed for target device to guarantee the convergence of proposed algorithm, it brings interference to other devices. 
So we focus on design interference-free scheduling by scheduling enrolled receiving devices as active ``PS" in different transmission block.
In this interference-free schedule scheme, there are no two enrolled receiving devices connected to the same device.
A naive scheduling policy is to schedule to the enrolled receiving devices one by one, which require $N$ transmission blocks during one consensus phase.
The graph coloring algorithm can be adopted to find a suitable scheduling policy that decreases the number of transmission blocks \cite{molloy2002graph}.
The authors in \cite{ozfatura2020decentralized} has investigated the effect of scheduling policy.
In this paper we focus on the convergence versus the consensus iteration $t$ over wireless networks.

\subsubsection{Transmission}
To enhance the spectral efficiency, AirComp has been envisioned to have a wide range of applications in the areas of consensus \cite{zhu2020over}.
In AirComp, the receiver device receives a superposition of the signals that simultaneously transmitted by its neighbor \cite{zhi2020air, dong2020air}.
Block flat-fading channel are considered in this paper and one transmission block contains $d$ time slots.
Hence, at the transmission block $t$, the received signal at enrolled receiving device $i$ can be written as 
\begin{align}
\bm{y}_i^t = \sum_{j \in \mathcal{N}_i} h_{ij}^t\bm{x}_j^t + \bm{z}_i^t,
\end{align}
where $h_{ij}^t \in \C$ is the channel coefficient between devices $i$ and $j$ at transmission block $t$, the transmit signal $\bm{x}_j^t$ encodes the information of the local model $\bm{\theta}_j^t$ and $\bm{z}_i^t \in 
\C^d$ represents the additive noise vector.
In addition, channel inputs are subject to a peak power constraint, i.e., $
\E[\Vert \bm{x}_i^t \Vert_2^2 ] \leq P, \forall i \in \mathcal{N}$.

We assume the computation phase at each device to be instantaneous after transmission.
Based on the received signal $\bm{y}_i^{t}$, device $i$ can average its neighbor's local model and get close to the solution of problem (\ref{eq: problem}).
However, due to the influence of channel fading and transmission noise, the received signal is distorted, which will negatively influence the convergence of DSGT-VR algorithms.
Hence, we aim to develop a reliable transmission strategy based on DSGT-VR to support FL in decentralized setting.

\subsection{AirComp-based DSGT-VR algortihm}
First, we develop precoding and decoding strategies at devices to guarantee the convergence of AirComp-based DSGT-VR.
Each consensus iteration $k$ includes $R \le N$ transmission blocks and the devices are scheduled as active ``PS" in one transmission block.  

In the precoding phase before AirComp, we assume that perfect CSI are available at all devices.
Since the topology of the communication network is known, each device contains the mixing matrix weights among its connected neighbors, i.e., $w_{ij}, \forall j \in \mathcal{N}_i$.
In transmission block $t$, the signal $\bm{x}_j^t, j \in \mathcal{N}_i$ transmitted to device $i$ is precoded as
\begin{align}\label{precoding}
\bm{x}_j^t = \sqrt{p^t} w_{ij} \frac{(h_{ij}^t)^H}{|h_{ij}^t|^2} \bm{\theta}_j^t,
\end{align}
where $\sqrt{p^t} = \min_{i} \frac{|h_{ij}^t| \sqrt{P}}{\|\bm{\theta}_i^t\|_2 }, \forall j$ is a uniform scaling factor to satisfy the peak power constraint.
However, in decentralized setting, it is observed that the weak channels lead to the small transmit power.
In this case, the influence of channel noise will significantly increase.
Therefore, we form the connectivity graph $\mathcal{G}$ by connecting the devices $i$ and $j$ if the channel gain $|h_{ij}|$ is above a certain threshold $\gamma$ \cite{vkn2011dynamic,vkn2012interfer,jafar2014tolo}.
Then, for device $i$, we get all its connected neighbors $\mathcal{N}_i$, which is given by
\begin{align}
\mathcal{N}_i = \{j| (i,j) \in \mathcal{E}, \text{if}~|h_{ij}| > \gamma\}.
\end{align}

Second, in the decoding phase and based on (\ref{precoding}), the received signal at device $i$ after AirComp can be written as
\begin{align}\label{Aircomp}
\bm{y}_i^t &= \sum_{j \in \mathcal{N}_i  \backslash \{i\}} \sqrt{p^t}  w_{ij} \bm{\theta}_j^t + \bm{z}_i^t. 
\end{align}
Then the device $i$ can decode the local model signal $\bm{y}_i^t$, by
\begin{align}\label{decoding}
\bm{\theta}_i^{t+1} &= \frac{1}{\sqrt{p^t}} \bm{y}_i^t + w_{ii}\bm{\theta}_i^t = \sum_{j \in \mathcal{N}_i} w_{ij} \bm{\theta}_j^t + \widetilde{\bm{z}}_i^t ,
\end{align}
where $ \widetilde{\bm{z}}_i^t = \frac{\bm{z}_i^t}{\sqrt{p}^t} \sim \mathcal{C}\mathcal{N}(0, \frac{\sigma^2}{p^t})$ is the channel noise.
Since equation (\ref{decoding}) is a superposition of the local model parameter, we apply AirComp to perform the consensus phases in DSGT-VR algorithm in wireless communication.

The AirComp techniques operate over star graphs in a pair of consecutive phases of time slots \cite{sery2020flmac,mo2020fl}.
However, in decentralized setting, since there is no central PS, each device can play a role of central PS through AirComp over star graphs \cite{kairouz2019advances}.
Specially, in the first phase of time slot, the device $i$ at the center of a star sub-graph receive the signals simultaneously transmitted by all its connected neighbors in $\mathcal{N}_n$ with a superposition form (\ref{Aircomp}).
In the second phase of time slot, this central device $i$ broadcasts its update gradient estimator (\ref{gradient_est}) to all its connected neighbors in $\mathcal{N}_n$.
In this paper, the devices broadcast the gradient estimator rather than model parameter \cite{xing2020decentralized,ozfatura2020decentralized} in downlink communication.
The complete implementation of AirComp-based DSGT-VR is summarized in Algorithm \ref{algorithm 2}.
Note that gradient estimator of each device is received in different transmission block via error-free link due to scheduling, then we can update gradient estimator based on (6) after $M$ transmission blocks.
And we leave the variance reduction techniques over noisy links for future work.

\begin{algorithm}[tbp]
  \caption{DSGT-VR with Over-the-Air Consensus} 
  \label{algorithm 2}
        \SetAlgoLined
        \SetKwInOut{Input}{Input}
        \SetKwInOut{Output}{Output}
        \SetKwFor{ParFor}{for each}{do in parallel}{end}
        \Input{Initial: $\bm{\theta}_i^0 = \bm{\theta}^0, \bm{d}_i^0 = \bm{g}_i^0 = \nabla f_{i}(\bm{\theta}_i^0), i \in \mathcal{N}$, step sizes $ \alpha$,  max iterations $T$, $\bm{W}$ and gradient table $\{ \nabla f_{i,j}(\hat{\bm{\theta}}_{i,j}) \}_{j=1}^{|\mathcal{D}_i|}, \hat{\bm{\theta}}_{i,j} = \bm{\theta}_i^0, \forall j$}
        \For{$t=1,2,\ldots,T$}{
        \For{$r = 1,2,\ldots, M$ transmission block}{
        \ParFor{$i \in \mathcal{N}$}{
        Updating $\bm{\theta}_i^{t+\frac{1}{2}} = \bm{\theta}_i^{t} - \alpha \bm{d}_i^t $.\\
        Get a data sample $\xi_n^{t+1}$ randomly, compute unbiased gradient (\ref{unbag}).\\
        Replace $\nabla f_{i,\xi_i^t}(\hat{\bm{\theta}}_{i, \xi_i^t})$ by $\nabla f_{i,\xi_i^t}(\bm{\theta}_i^t)$ in gradient table.\\
        Updating $\bm{d}_i^{t+\frac{1}{2}} = \bm{d}_i^{t} + \bm{g}_{i}^{t} - \bm{g}_{i}^{t-1}$.\\
        \uIf{device $i$ is scheduled as active ``PS"}{
        The device $i$ received signal via (\ref{Aircomp}) and decoding via (\ref{decoding}) to get $\bm{\theta}_i^{t+1}$.\\
        Broadcast $\bm{d}_i^{t+\frac{1}{2}}$ back to all its neighbors.\\
        }
        \Else{Precoding $\bm{\theta}_i^{t+\frac{1}{2}} $ via (\ref{precoding}) and send the precoded signal to all its neighbors $\mathcal{N}_i$.\\}
        }
        }
        Each device locally updates gradient eatimator $\bm{d}_i^{t+ 1} = \sum_{j \in \mathcal{N}_i} w_{ij}\bm{d}_i^{t+\frac{1}{2}}$. 
        } 
\end{algorithm}

\section{Convergence Analysis}
In this section, we provide the convergence analysis of the AirComp-based DSGT-VR.
The convergence results are established under the following assumptions.
\begin{assumption}\label{ass:1}
The local cost function $f_i$ is both $\mu$-strongly convex and $L$-smooth for any $i \in \mathcal{N}$.
\end{assumption}
\begin{assumption}\label{ass:2}
The network graph $\mathcal{G}$ is undirected and connected, i.e., there exists a path between any two devices.
\end{assumption}
\begin{assumption}
The local parameter is bounded by a universal constant $B \ge 0$, i.e., $\| \bm{\theta}_i^t \|^2 \leq B^2, \forall i,t$.
\end{assumption}
Assumption \ref{ass:2} implies that $\beta = \interleave \bm{W} - \frac{1}{N} \bm{1}_N \bm{1}_N^\top \interleave < 1$ \cite{xin2020GTSAGA}.
Based on above assumptions, we denote $M = \max_i |\mathcal{D}_i|, m = \min_i |\mathcal{D}_i|, \forall i \in \mathcal{N}$ and $\kappa = \frac{L}{\mu}$ be the condition number of the global loss function $F$.
The main convergence result of AirComp-based DSGT-VR is established as follows,
\begin{theorem}\label{thm: theorem 1}
Let $\bm{\theta}^\natural$ denote the solution of the distributed optimization problem (\ref{eq: problem}).
If Assumptions \ref{ass:1} and \ref{ass:2} hold and the step-size in AirComp-based DSGT-VR algorithm is such that
\begin{align}
\alpha = \min\{ \mathcal{O}(\frac{1}{\mu M}), \mathcal{O}(\frac{m(1-\beta)^2}{ML \kappa^2}) \} \notag ,
\end{align}
then $\forall t \geq 0 , \forall i \in \mathcal{N}$ and for some $0 < \rho <1$, it holds that
\begin{align}
\E\left[F(\bm{\theta}_i^{t})\right] - F(\bm{\theta}^\natural) \leq \frac{cL}{2} \rho^t + \frac{LN}{2(N-1)} \frac{d \sigma^2 B^2}{\gamma^2 P} \sum_{\tau=1}^{t} \rho^{t - \tau},
\end{align}
where $ c = \frac{N}{N-1} \| \bm{\theta}_i^0 - \overline{\bm{\theta}}^0 \|^2  + N \| \overline{\bm{\theta}}^0 - \bm{\theta}^{\natural} \|^2  $.
\end{theorem}
\begin{proof}
See Appendix \ref{sec: proof of theorem 1}.
\end{proof}
Theorem \ref{thm: theorem 1} establishes the upper bound of estimation error for strongly convex and smooth local loss function in decentralized federated learning system.
The error bound is divided into the initial distance due to the error in original algorithm and the additive noise caused by the channel noise.
In addition, the initial distance shows that the AirComp-based DSGT-VR achieves the same convergence rate as that of the DSGT-VR algorithm, i.e., linear convergence rate.
Theorem \ref{thm: theorem 1} demonstrates that impact of the additive noise decreases as the number of devices increases. 

\section{Numerical Experiments}
\begin{figure*}[tbp]
        \centering
        \begin{minipage}{.32\textwidth}
                \centering
                \includegraphics[width=1\linewidth]{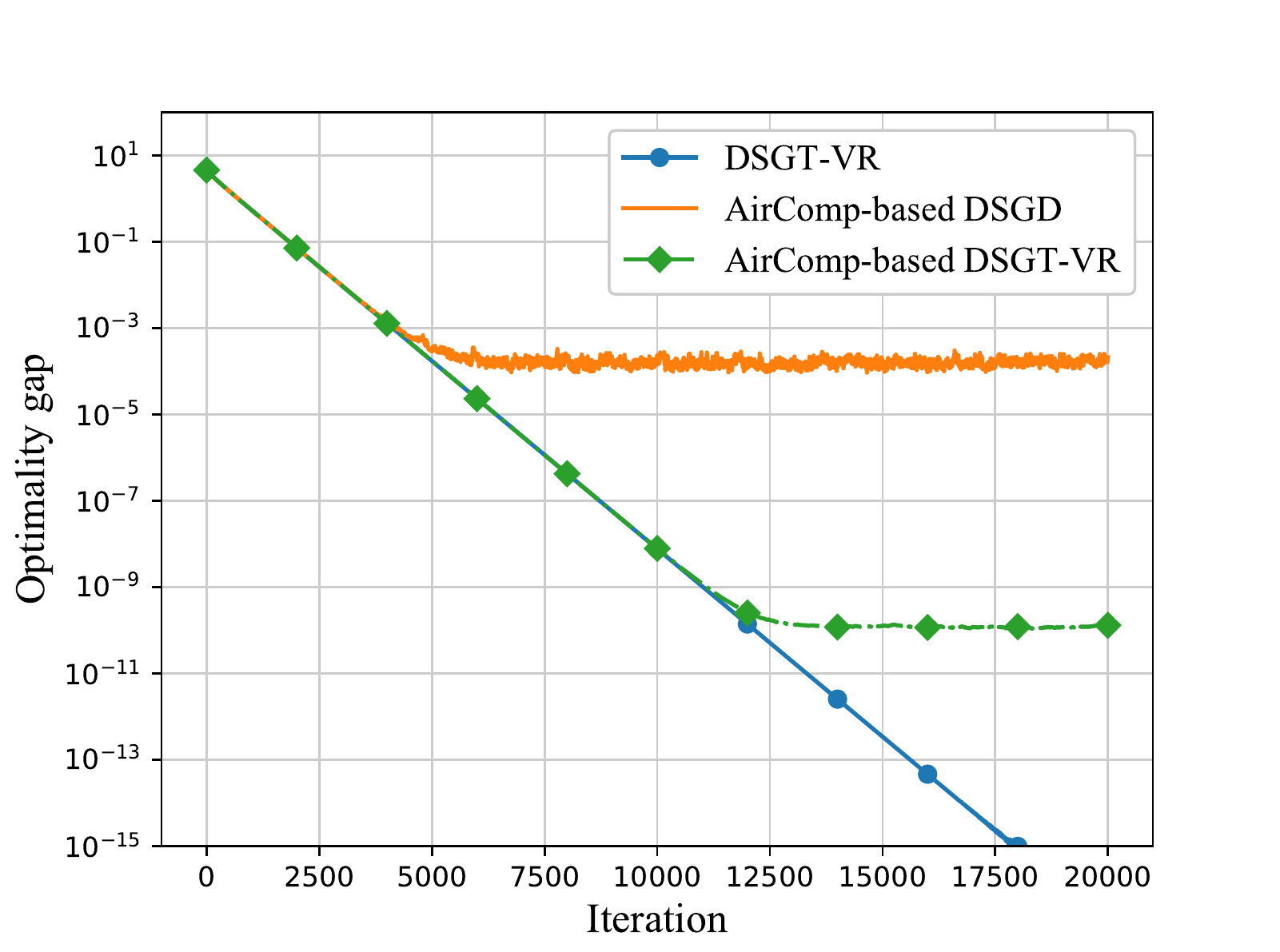}
                \caption{Optimality gap versus number \\ of consensus iteration $t$.}
                \label{fig:1}
        \end{minipage}
        \vspace{-2mm}
        \begin{minipage}{.32\textwidth}
                \centering
                \includegraphics[width=1\linewidth]{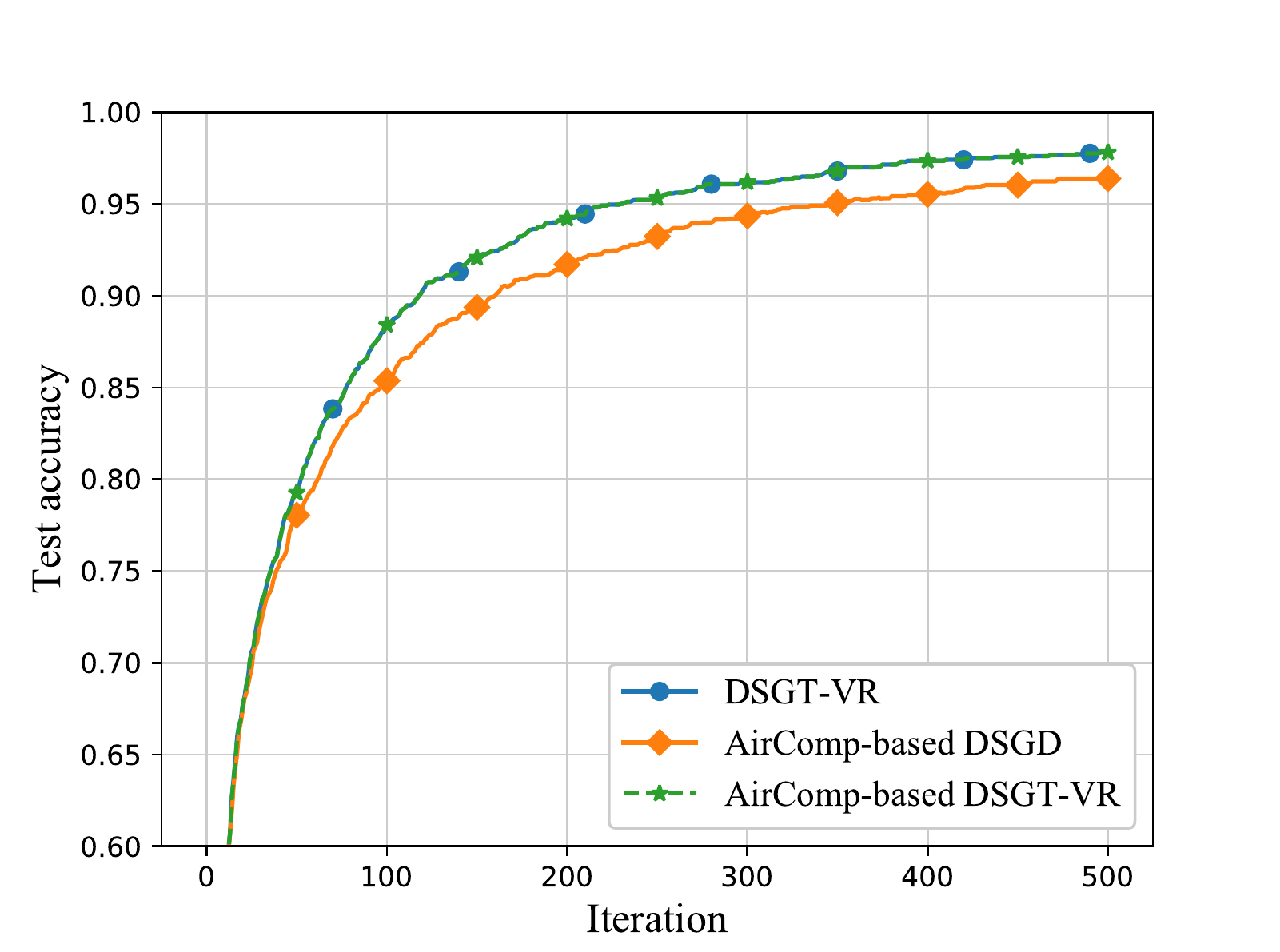}
                \caption{Test accuracy versus number \\ of consensus iteration $t$.}
                \label{fig:2}
        \end{minipage}
        \vspace{-2mm}
        \begin{minipage}{.32\textwidth}
                \centering
                \includegraphics[width=1\linewidth]{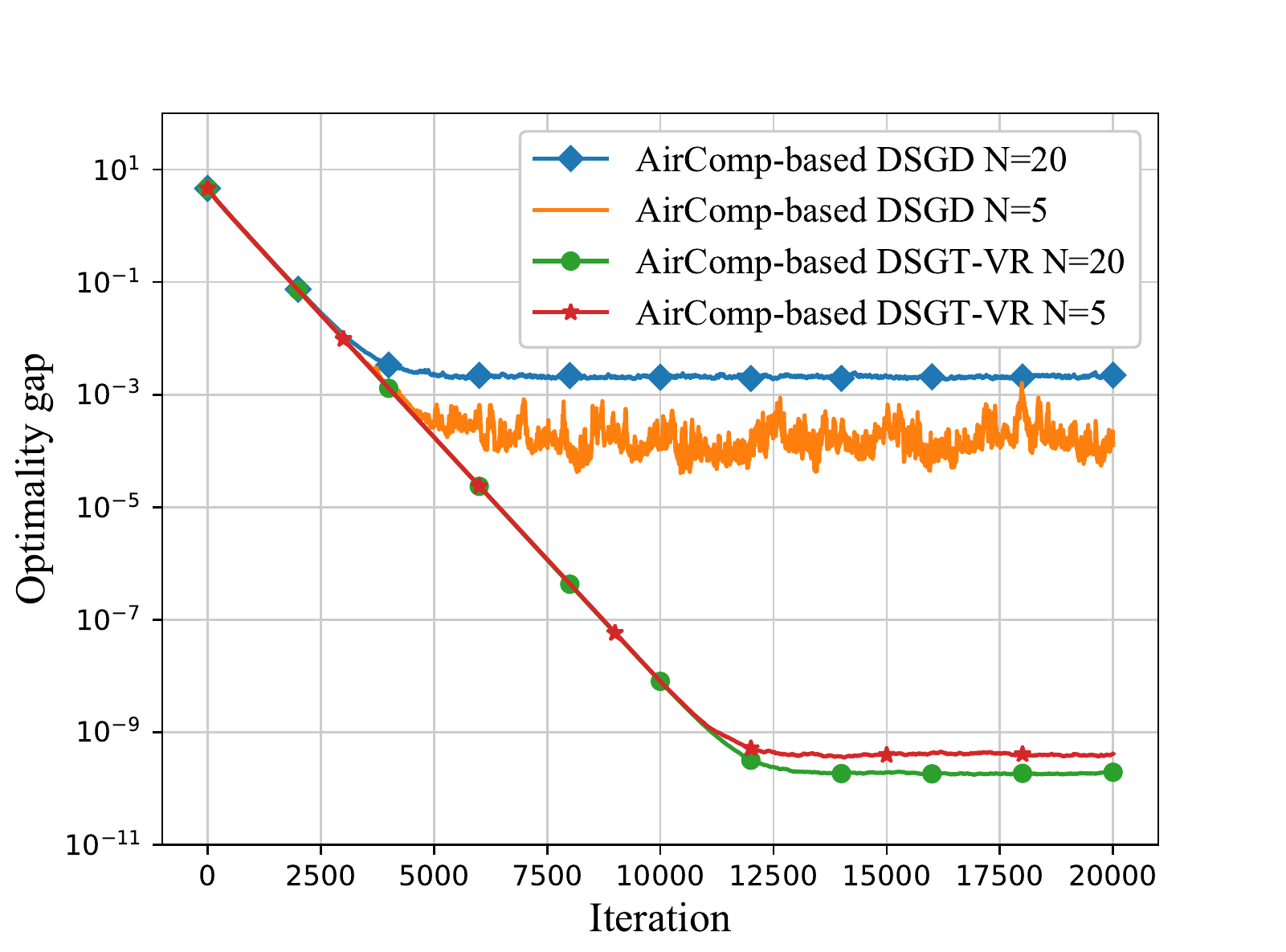}
                \caption{Optimality gap versus number \\ of consensus iteration $t$ with different $N$.}
                \label{fig:3}
        \end{minipage}
        \vspace{-2mm}
\end{figure*}
In this section, we numerically demonstrate the convergence behavior of the proposed AirComp-based DSGT-VR, and compare with the-state-of-art decentralized FL algorithms.
We consider a decentralized federated learning setting where $N = 20$ devices with $1000$ data samples locally at each device cooperatively train a regularized logistic regression model for binary classification,
\begin{align}
F(\bm{\theta}) = \frac{1}{N} \sum_{i=1}^{N} \frac{1}{m_i} \sum_{j=1}^{m_i} \log[1 + e^{-(\bm{a}_{ij}^T \bm{\theta})b_{ij}}] + \frac{\lambda}{2} \Vert \bm{x} \Vert_2^2, 
\end{align}
where device $i$ holds $m_i$ training samples $\{\bm{a}_{ij}, b_{ij}\}_{j=1}^{m_i} $, $\bm{a}_{ij}$ denotes the feature of $j$-th training sample at $i$-th device, $b_{ij} \in \{+1, -1\} $ is the corresponding
binary label and $\lambda$ is a regularization parameter.

For data sample, we use classes ``3" and ``5" in the MNIST dataset for binary classification.
We use $1000$ samples for training and $1968$ samples for testing, and each device is assigned by $40$ training samples.
We use the channel $h_{ij}^t \sim \mathcal{C}\mathcal{N}(0,1) $, noise $\sigma^2$ to be $ 0$ dBm and the threshold $\gamma = 0.5$.
In this paper, we adopted Laplacian matrix \cite{xiao2004fast} as mixing matrix for fast and provable convergence analysis. 
To stable the training process, all features are normalized to be unit vectors, i.e., $\Vert \bm{a}_{ij} \Vert_2 = 1, \forall i,j$, and set the regularization parameter $\lambda = \frac{1}{\sum_{i=1}^{N} m_i}$.
And we characterize the performance in terms of the optimality gap, i.e., $\E\left[F(\bm{\theta}_i^{t})\right] - F(\bm{\theta}^\natural)$.
To further verify the effectiveness of the proposed, we compare the proposed algorithm with the AirComp-based DSGD algorithm \cite{ozfatura2020decentralized, xing2020decentralized}, which is the state-of-the-art algorithm of the wireless decentralized federated edge learning system.

Fig. \ref{fig:1} shows that both AirComp-based DSGD and the proposed algorithm achieve a linear convergence rate, but the proposed algorithm can converge to a more accurate solution.
Specially, the error gap of AirComp-based DSGD almost reaches $10^{-4} $, while the proposed algorithm only has the optimality gap of $10^{-10}$.
However, due to the effect by fading channels and additive noise, there is still a gap between the solution obtained by proposed algorithm and optimal solution, which corresponds to our theoretical results.
And the test accuracy is evaluated versus the number of consensus iteration is illustrated in Fig. \ref{fig:2}.
Obviously, the proposed algorithm outperforms the AirComp-based DSGD algorithm while reaches the nearly the same accuracy of DSGT-VR algorithm.

In Fig \ref{fig:3}, we illustrate the impact of the different number of devices on optimality gap of the proposed algorithm in the same network topology.
Clearly, the proposed algorithm with $20$ devices reaches a more accurate solution than that with $5$ devices, which verifies our theoretical results.
However, the AirComp-based DSGD has the opposite result that the more devices the less accurate solution.

\section{Conclusions}
In this paper, we proposed a novel AirComp-based DSGT-VR algorithm to facilitate the consensus phase in decentralized FL over wireless networks, where both precoding and decoding strategies at devices are developed for D2D communications.
Furthermore, we proved the linear convergence rate of the proposed algorithm and provided the error bound to reveal the impact of fading channel and its noise.
Simulation results were conducted to verify the theoretical result and the superior performance of the proposed algorithm.
\vspace{-2mm}

\appendices
\section{Proof of Theorem \ref{thm: theorem 1}} \label{sec: proof of theorem 1}
Under assumption \ref{ass:1}, the global cost $F$ is also $L$-smooth and we denote $\bm{\theta}^\natural$ is the optimal solution of $F$, then by smoothness
\begin{align}\label{smooth}
\E\left[F(\bm{\theta}_i^{t})\right] - F(\bm{\theta}^\natural) \leq \frac{L}{2}\E \left[ \| \bm{\theta}_i^t - \bm{\theta}^{\natural} \|^2\right].
\end{align}
First, we introduce the auxiliary variable $\{ \overline{\bm{\theta}}^t \}$ which is mean of $\bm{\theta}_1^t, \cdots, \bm{\theta}_N^t$ in error-free case, then the error $\E\left[ \| \bm{\theta}_i^t - \bm{\theta}^{\natural} \|^2 \right], \forall i \in \mathcal{N}$ turn out to be
\begin{align}\notag
&\E\left[ \| \bm{\theta}_i^t - \bm{\theta}^{\natural} \|^2 \right] = \E\left[ \| \bm{\theta}_i^t - \overline{\bm{\theta}}^t + \overline{\bm{\theta}}^t - \bm{\theta}^{\natural} \|^2 \right] \\
&\overset{(\mathrm{a})}{\leq} \underbrace{ \frac{N}{N-1}\E\left[ \| \bm{\theta}_i^t - \overline{\bm{\theta}}^t \|^2 \right] }_{T_1} + \underbrace{ N\E \left[ \| \overline{\bm{\theta}}^t - \bm{\theta}^{\natural} \|^2 \right]  }_{T_2},
\end{align}
where $(\mathrm{a})$ comes from Young's inequality that $\| \bm{a} + \bm{b} \|^2 \leq (1+ \eta)\| \bm{a}\|^2 + (1 + \frac{1}{\eta})  \| \bm{b}\|^2, \forall \bm{a}, \bm{b} \in \R^{d} , \forall \eta \ge 0$ and we set $\eta = N-1$.
Then the error  $\E\left[ \| \bm{\theta}_i^t - \bm{\theta}^{\natural} \|^2 \right]$ is divided into two parts, i.e., $T_1$ and $T_2$.
Next, we focus on establish the upper bounds for $T_1$ and $T_2$.
Recall (\ref{decoding}), we have $\bm{\theta}_i^t = \sum_{j \in \mathcal{N}_i} w_{ij} \bm{\theta}_j^{t-1} + \widetilde{\bm{z}}_i^{t-1}$, then we can rewrite $\E\left[ \| \bm{\theta}_i^t - \overline{\bm{\theta}}^t \|^2 \right]$ in $T_1$ as follows
\begin{align*} \label{T1_1}
&\E\left[ \| \bm{\theta}_i^t - \overline{\bm{\theta}}^t \|^2 \right] 
\overset{(\mathrm{a})}{=} \E\left[ \| \sum_{j \in \mathcal{N}_i} w_{ij} \bm{\theta}_j^{t-1}- \overline{\bm{\theta}}^t \|^2 \right] + \E\left[ \| \widetilde{\bm{z}}_i^{t-1} \|^2 \right], \\
&\overset{(\mathrm{b})}{\leq} \rho \E\left[ \| \bm{\theta}_i^{t-1}- \overline{\bm{\theta}}^{t-1} \|^2 \right] + \frac{d \sigma^2 B^2}{\gamma^2 P},
\end{align*}
where $(\mathrm{a})$ comes from that noise $\widetilde{\bm{z}}_i^{t-1}$ is zero-mean and independent to model parameters and $(\mathrm{b})$ comes from \cite[Theorem 1, Lemma 2]{qu2017harnessing} with $0 < \rho < 1$ and the uniform scaling factor $\sqrt{p^t}$, i.e., $
\E\left[ \| \widetilde{\bm{z}}_i^{t-1} \|^2 \right] = \frac{d \sigma^2}{p^t} \leq \frac{d \sigma^2 B^2}{\gamma^2 P}$.
According to \cite[Propositon 1]{xin2020GTSAGA}, then if $\alpha = \min\{ \mathcal{O}(\frac{1}{\mu M}), \mathcal{O}(\frac{m(1-\beta)^2}{ML \kappa^2}) \}$, we have $ T_2 \le \rho \E \left[ N\| \overline{\bm{\theta}}^{t-1} - \bm{\theta}^{\natural} \|^2 \right] $.
Hence, we get
\begin{align*}
\E\left[ \| \bm{\theta}_i^t - \bm{\theta}^{\natural} \|^2 \right] 
&\le c \rho^t + \frac{N}{N-1} \frac{d \sigma^2 B^2}{\gamma^2 P} \sum_{\tau=1}^{t} \rho^{t - \tau},
\end{align*}
where $ c = \frac{N}{N-1} \| \bm{\theta}_i^0 - \overline{\bm{\theta}}^0 \|^2  + N \| \overline{\bm{\theta}}^0 - \bm{\theta}^{\natural} \|^2  $ and $0 < \rho <1$.
Based on (\ref{smooth}), we can obtain
\begin{align}
\E\left[F(\bm{\theta}_i^{t})\right] - F(\bm{\theta}^\natural) \leq \frac{cL}{2} \rho^t + \frac{LN}{2(N-1)} \frac{d \sigma^2 B^2}{\gamma^2 P} \sum_{\tau=1}^{t} \rho^{t - \tau}  ,
\end{align}
which completes the proof.

\bibliographystyle{IEEEbib}
\bibliography{strings,refs}

\begin{thebibliography}{10}

\bibitem{mcmahan2017communication}
B.~{McMahan}, E.~{Moore}, D.~{Ramage}, S.~{Hampson}, and B.~A. {Arcas},
\newblock ``Communication-efficient learning of deep networks from
  decentralized data,''
\newblock in {\em Proc. Int. Conf. Artif. Intell. and Statist. (AISTATS)}.
  PMLR, 2017, pp. 1273--1282.

\bibitem{rah2020survey}
S.~A. {Rahman}, H.~{Tout}, H.~{Ould-Slimane}, A.~{Mourad}, C.~{Talhi}, and
  M.~{Guizani},
\newblock ``A survey on federated learning: The journey from centralized to
  distributed on-site learning and beyond,''
\newblock {\em IEEE Internet of Things J.}, Apr. 2021.

\bibitem{yuan2020edgeai}
Y.~{Shi}, K.~{Yang}, T.~{Jiang}, J.~{Zhang}, and K.~B. {Letaief},
\newblock ``Communication-efficient edge ai: Algorithms and systems,''
\newblock {\em IEEE Commun. Surveys Tuts.}, vol. 22, no. 4, pp. 2167--2191,
  Oct. 2020.

\bibitem{lu2020iiot}
Y.~{Lu}, X.~{Huang}, Y.~{Dai}, S.~{Maharjan}, and Y.~{Zhang},
\newblock ``Blockchain and federated learning for privacy-preserved data
  sharing in industrial {IoT},''
\newblock {\em IEEE Trans. Industr. Inform.}, vol. 16, no. 6, pp. 4177--4186,
  Jun. 2020.

\bibitem{pham2021fusion}
Q.~V. {Pham}, K.~{Dev}, P.~K.~R. {Maddikunta}, T.~R. {Gadekallu},
  T.~{Huynh-The}, et~al.,
\newblock ``Fusion of federated learning and industrial internet of things: A
  survey,''
\newblock {\em arXiv preprint arXiv:2101.00798}, 2021.

\bibitem{savazzi2021opportunities}
S.~{Savazzi}, M.~{Nicoli}, M.~{Bennis}, S.~{Kianoush}, and L.~{Barbieri},
\newblock ``Opportunities of federated learning in connected, cooperative, and
  automated industrial systems,''
\newblock {\em IEEE Commun. Mag.}, Mar. 2021.

\bibitem{yang2021federated}
Z.~{Yang}, M.~{Chen}, K.~{Wong}, H.~{Poor}, and S.~{Cui},
\newblock ``Federated learning for {6G}: Applications, challenges, and
  opportunities,''
\newblock {\em arXiv preprint arXiv:2101.01338}, 2021.

\bibitem{lim2020fl}
W.~Y.~B. {Lim}, N.~C. {Luong}, D.~T. {Hoang}, Y.~{Jiao}, Y.~C. {Liang},
  Q.~{Yang}, D.~{Niyato}, and C.~{Miao},
\newblock ``Federated learning in mobile edge networks: A comprehensive
  survey,''
\newblock {\em IEEE Commun. Surveys Tuts.}, vol. 22, no. 3, pp. 2031--2063,
  Apr. 2020.

\bibitem{nik2020fl}
S.~{Niknam}, H.~S. {Dhillon}, and J.~H. {Reed},
\newblock ``Federated learning for wireless communications: Motivation,
  opportunities, and challenges,''
\newblock {\em IEEE Commun. Mag.}, Jul. 2020.

\bibitem{kb20196g}
K.~B. {Letaief}, W.~{Chen}, Y.~{Shi}, J.~{Zhang}, and Y.~A. {Zhang},
\newblock ``The roadmap to {6G}: {AI} empowered wireless networks,''
\newblock {\em IEEE Commun. Mag.}, vol. 57, no. 8, pp. 84--90, Aug. 2019.

\bibitem{yang2020flris}
K.~{Yang}, Y.~{Shi}, Y.~{Zhou}, Z.~{Yang}, L.~{Fu}, and W.~{Chen},
\newblock ``Federated machine learning for intelligent {IoT} via reconfigurable
  intelligent surface,''
\newblock {\em IEEE Netw.}, vol. 34, no. 5, pp. 16--22, Oct. 2020.

\bibitem{li2021delay}
L.~{Li}, L.~{Yang}, X.~{Guo}, Y.~{Shi}, H.~{Wang}, W.~{Chen}, and K.~B.
  {Letaief},
\newblock ``Delay analysis of wireless federated learning based on saddle point
  approximation and large deviation theory,''
\newblock {\em arXiv preprint arXiv:2103.16994}, 2021.

\bibitem{wang2020federated}
Z.~{Wang}, J.~{Qiu}, Y.~{Zhou}, Y.~{Shi}, L.~{Fu}, W.~{Chen}, and K.~B.
  {Lataief},
\newblock ``Federated learning via intelligent reflecting surface,''
\newblock {\em arXiv preprint arXiv:2011.05051}, 2020.

\bibitem{chen2020conver}
M.~{Chen}, H.~{Vincent Poor}, W.~{Saad}, and S.~{Cui},
\newblock ``Convergence time optimization for federated learning over wireless
  networks,''
\newblock {\em IEEE Trans. Wireless Commun.}, vol. 20, no. 4, pp. 2457--2471,
  Apr. 2021.

\bibitem{chen2020joint}
M.~{Chen}, Z.~{Yang}, W.~{Saad}, C.~{Yin}, H.~V. {Poor}, and S.~{Cui},
\newblock ``A joint learning and communications framework for federated
  learning over wireless networks,''
\newblock {\em IEEE Trans. Wireless Commun.}, vol. 20, no. 1, pp. 269--283,
  Jan. 2021.

\bibitem{wen2020joint}
D.~{Wen}, M.~{Bennis}, and K.~{Huang},
\newblock ``Joint parameter-and-bandwidth allocation for improving the
  efficiency of partitioned edge learning,''
\newblock {\em IEEE Trans. Wireless Commun.}, vol. 19, no. 12, pp. 8272--8286,
  Dec. 2020.

\bibitem{yang2020over}
K.~{Yang}, T.~{Jiang}, Y.~{Shi}, and Z.~{Ding},
\newblock ``Federated learning via over-the-air computation,''
\newblock {\em IEEE Trans. Wireless Commun.}, vol. 19, no. 3, pp. 2022--2035,
  Mar. 2020.

\bibitem{sery2020flmac}
T.~{Sery} and K.~{Cohen},
\newblock ``On analog gradient descent learning over multiple access fading
  channels,''
\newblock {\em IEEE Trans. Signal Process.}, vol. 68, pp. 2897--2911, Apr.
  2020.

\bibitem{mo2020fl}
M.~{Mohammadi Amiri} and D.~{Gündüz},
\newblock ``Machine learning at the wireless edge: Distributed stochastic
  gradient descent over-the-air,''
\newblock {\em IEEE Trans. Signal Process.}, vol. 68, pp. 2155--2169, Mar.
  2020.

\bibitem{lian2017can}
X.~{Lian}, C.~{Zhang}, H.~{Zhang}, C.~{Hsieh}, W.~{Zhang}, and J.~{Liu},
\newblock ``Can decentralized algorithms outperform centralized algorithms? a
  case study for decentralized parallel stochastic gradient descent,''
\newblock in {\em Proc. Neural Inf. Process. Syst. (NeurIPS)}, 2017, pp.
  5330--5340.

\bibitem{cai2020leter}
X.~{Cai}, X.~{Mo}, J.~{Chen}, and J.~{Xu},
\newblock ``{D2D}-enabled data sharing for distributed machine learning at
  wireless network edge,''
\newblock {\em IEEE Wireless Commun. Lett.}, vol. 9, no. 9, pp. 1457--1461,
  Sept. 2020.

\bibitem{ozfatura2020decentralized}
E.~Ozfatura, Stefano Rini, and D.~Gündüz,
\newblock ``Decentralized {SGD} with over-the-air computation,''
\newblock in {\em Proc. IEEE Global Commun. Conf. (Globecom)}, 2020, pp. 1--6.

\bibitem{xing2020decentralized}
H.~{Xing}, O.~{Simeone}, and S.~{Bi},
\newblock ``Decentralized federated learning via {SGD} over wireless {D2D}
  networks,''
\newblock in {\em Proc. IEEE Int. Workshop Signal Process. Adv. Wireless
  Commun. (SPAWC)}, 2020, pp. 1--5.

\bibitem{sa2020de}
S.~{Savazzi}, S.~{Kianoush}, V.~{Rampa}, and M.~{Bennis},
\newblock ``A joint decentralized federated learning and communications
  framework for industrial networks,''
\newblock in {\em Proc. IEEE Int. Workshop Comput.-Aided Model., Anal., and
  Dess. Commun. Links and Netw. (CAMAD)}, 2020, pp. 1--7.

\bibitem{kairouz2019advances}
P.~{Kairouz}, H.~B. {McMahan}, B.~{Avent}, A.~{Bellet}, M.~{Bennis}, A.~N.
  {Bhagoji}, K.~{Bonawitz}, Z.~{Charles}, G.~{Cormode}, R.~{Cummings}, et~al.,
\newblock ``Advances and open problems in federated learning,''
\newblock {\em arXiv preprint arXiv:1912.04977}, 2019.

\bibitem{koloskova2020unified}
A.~{Koloskova}, N.~{Loizou}, S.~{Boreiri}, M.~{Jaggi}, and S.~U. {Stich},
\newblock ``A unified theory of decentralized sgd with changing topology and
  local updates,''
\newblock {\em arXiv preprint arXiv:2003.10422}, 2020.

\bibitem{xin2020magdso}
R.~{Xin}, S.~{Kar}, and U.~A. {Khan},
\newblock ``Decentralized stochastic optimization and machine learning: A
  unified variance-reduction framework for robust performance and fast
  convergence,''
\newblock {\em IEEE Signal Process. Mag.}, vol. 37, no. 3, pp. 102--113, May
  2020.

\bibitem{pu2020distributed}
S.~{Pu} and A.~{Nedi{\'c}},
\newblock ``Distributed stochastic gradient tracking methods,''
\newblock {\em Math. Program.}, pp. 1--49, 2020.

\bibitem{te2014d2d}
M.~N. {Tehrani}, M.~{Uysal}, and H.~{Yanikomeroglu},
\newblock ``Device-to-device communication in {5G} cellular networks:
  challenges, solutions, and future directions,''
\newblock {\em IEEE Commun. Mag.}, vol. 52, no. 5, pp. 86--92, May 2014.

\bibitem{molloy2002graph}
M.~S. {Molloy}, M.~{Molloy}, and B.~{Reed},
\newblock {\em Graph colouring and the probabilistic method}, vol.~23,
\newblock Springer Science \& Business Media, 2002.

\bibitem{zhu2020over}
G.~{Zhu}, J.~{Xu}, and K.~{Huang},
\newblock ``Over-the-air computing for {6G}--turning air into a computer,''
\newblock {\em arXiv preprint arXiv:2009.02181}, 2020.

\bibitem{zhi2020air}
Z.~{Wang}, Y.~{Shi}, Y.~{Zhou}, H.~{Zhou}, and N.~{Zhang},
\newblock ``Wireless-powered over-the-air computation in intelligent reflecting
  surface-aided {IoT} networks,''
\newblock {\em IEEE Internet of Things J.}, vol. 8, no. 3, pp. 1585--1598, Feb.
  2021.

\bibitem{dong2020air}
J.~{Dong}, Y.~{Shi}, and Z.~{Ding},
\newblock ``Blind over-the-air computation and data fusion via provable
  wirtinger flow,''
\newblock {\em IEEE Trans. Signal Process.}, vol. 68, pp. 1136--1151, Jan.
  2020.

\bibitem{vkn2011dynamic}
L.~{Ruan} and V.~K.~N. {Lau},
\newblock ``Dynamic interference mitigation for generalized partially connected
  quasi-static {MIMO} interference channel,''
\newblock {\em IEEE Trans. Signal Process.}, vol. 59, no. 8, pp. 3788--3798,
  Aug. 2011.

\bibitem{vkn2012interfer}
L.~{Ruan}, V.~K.~N. {Lau}, and X.~{Rao},
\newblock ``Interference alignment for partially connected {MIMO} cellular
  networks,''
\newblock {\em IEEE Trans. Signal Process.}, vol. 60, no. 7, pp. 3692--3701,
  Jul. 2012.

\bibitem{jafar2014tolo}
S.~A. {Jafar},
\newblock ``Topological interference management through index coding,''
\newblock {\em IEEE Trans. Inf. Theory}, vol. 60, no. 1, pp. 529--568, Oct.
  2014.

\bibitem{xin2020GTSAGA}
R.~{Xin}, U.~A. {Khan}, and S.~{Kar},
\newblock ``Variance-reduced decentralized stochastic optimization with
  accelerated convergence,''
\newblock {\em IEEE Trans. Signal Process.}, vol. 68, pp. 6255--6271, Oct.
  2020.

\bibitem{xiao2004fast}
L.~{Xiao} and S.~{Boyd},
\newblock ``Fast linear iterations for distributed averaging,''
\newblock {\em Systems \& Control Letters}, vol. 53, no. 1, pp. 65--78, 2004.

\bibitem{qu2017harnessing}
G.~{Qu} and N.~{Li},
\newblock ``Harnessing smoothness to accelerate distributed optimization,''
\newblock {\em IEEE Trans. Control. Netw. Syst.}, vol. 5, no. 3, pp.
  1245--1260, Sept. 2018.

\end{thebibliography}

\end{document}